%% file: main.tex
\documentclass[11pt]{article}
\usepackage{setspace}
\usepackage{fullpage}
\usepackage{microtype}
\usepackage{graphicx}
\usepackage{subfigure}
\usepackage{booktabs} 
\usepackage[colorlinks,citecolor=blue]{hyperref}
\usepackage{multicol}
\usepackage{multirow}
\usepackage{makecell}
\usepackage{amsmath, amssymb, amsthm, amsfonts}
\usepackage{xcolor}
\PassOptionsToPackage{algo2e,ruled}{algorithm2e}
\usepackage{algorithm2e}
\usepackage{fullpage}
\usepackage[mathscr]{euscript}
\usepackage{nicefrac}
\usepackage{dsfont}
\usepackage{wrapfig}
\usepackage{xspace}
\usepackage{enumitem}
\usepackage{ifdraft}
\usepackage{authblk}
\usepackage[loadshadowlibrary,colorinlistoftodos,bordercolor=white]{todonotes}
\newcommand\addauthors[4]{
    \expandafter\providecommand\csname #1\endcsname[1]{{\addcontentsline{tdo}{todo}{\color{#4} (\textbf{#2}| ##1)}{\color{#4} (\textbf{#2}| ##1)}}{}}
    \expandafter\providecommand\csname #1l\endcsname[1]{\todo[color=#3,inline]{{\textbf{(#2)}} ##1}}
}
\usepackage{format/color-edits}
\addauthors{nacomment}{naman}{blue}{blue}
\addauthors{ks}{karan}{orange!20!white}{orange!60!black}
\addauthors{skcomment}{satyen}{green!40!black}{green!40!black}
\let\main\relax
\input{format/defs}

\usepackage{natbib}
\title{Improved Differentially Private and Lazy Online Convex Optimization}

\author[1]{Naman Agarwal\thanks{\texttt{namanagarwal@google.com}}}
\author[2]{Satyen Kale\thanks{\texttt{satyenkale@google.com}}}
\author[3]{Karan Singh\thanks{\texttt{karansingh@cmu.edu}}}
\author[1]{Abhradeep Thakurta\thanks{\texttt{athakurta@google.com}}}
\affil[1]{Google Deepmind}
\affil[2]{Google Research}
\affil[3]{Tepper School of Business, Carnegie Mellon University}

\begin{document}

\maketitle

\input{abstract}

\input{intro}

\input{problem}
\input{result}

\bibliography{main.bib}  
\appendix
\input{prelim}

\input{density-smoothness-proofs.tex}
\input{general-analysis.tex}

\input{privacy}

\end{document}

%% file: format/defs.tex
\numberwithin{equation}{section}

\DeclareMathOperator*{\argmin}{argmin}
\def\reals{{\mathbb R}}

\newcommand{\defeq}{\triangleq}

\newcommand{\barmu}{\bar{\mu}}

\newcommand{\Ber}{\mathrm{Ber}}

\ifdefined\main
\newtheorem{theorem}{Theorem}[section]
\newtheorem{claim}[theorem]{Claim}

\newtheorem{lemma}[theorem]{Lemma}

\newtheorem{definition}[theorem]{Definition}

\else

\newtheorem{claim}[theorem]{Claim}

\fi

\newtheorem{assumption}[theorem]{Assumption}

\usepackage{enumitem}

\newcommand{\replabel}{\label} 

\ExplSyntaxOn

\NewDocumentCommand{\repeattheorem}{m}
 {
  \group_begin:
  \renewcommand{\thetheorem}{\ref{#1}}
  \renewcommand{\replabel}[1]{\tag{\ref{##1}}}
  \prop_item:Nn \g_reptheorem_prop { #1 }
  \endtheorem
  \group_end:
 }

\NewDocumentEnvironment{reptheorem}{m+b}
 {
  \prop_gput:Nnn \g_reptheorem_prop { #1 } { \theorem #2 \endtheorem }
  \theorem#2\unskip\label{#1}\endtheorem
 }{}

\prop_new:N \g_reptheorem_prop

\NewDocumentCommand{\repeatlemma}{m}
 {
  \group_begin:
  \renewcommand{\thelemma}{\ref{#1}}
  \renewcommand{\replabel}[1]{\tag{\ref{##1}}}
  \prop_item:Nn \g_replemma_prop { #1 }
  \endlemma
  \group_end:
 }

\NewDocumentEnvironment{replemma}{m+b}
 {
  \prop_gput:Nnn \g_replemma_prop { #1 } { \lemma #2 \endlemma }
  \lemma#2\unskip\label{#1}\endlemma
 }{}

\prop_new:N \g_replemma_prop

\ExplSyntaxOff

\usepackage{prettyref}
\newcommand{\pref}[1]{\prettyref{#1}}

\newcommand{\savehyperref}[2]{\texorpdfstring{\hyperref[#1]{#2}}{#2}}
\newrefformat{eq}{\savehyperref{#1}{\textup{(\ref*{#1})}}}
\newrefformat{eqn}{\savehyperref{#1}{Equation~\ref*{#1}}}
\newrefformat{lem}{\savehyperref{#1}{Lemma~\ref*{#1}}}
\newrefformat{def}{\savehyperref{#1}{Definition~\ref*{#1}}}
\newrefformat{thm}{\savehyperref{#1}{Theorem~\ref*{#1}}}
\newrefformat{cor}{\savehyperref{#1}{Corollary~\ref*{#1}}}
\newrefformat{sec}{\savehyperref{#1}{Section~\ref*{#1}}}
\newrefformat{app}{\savehyperref{#1}{Appendix~\ref*{#1}}}
\newrefformat{ass}{\savehyperref{#1}{Assumption~\ref*{#1}}}
\newrefformat{ex}{\savehyperref{#1}{Example~\ref*{#1}}}
\newrefformat{fig}{\savehyperref{#1}{Figure~\ref*{#1}}}
\newrefformat{alg}{\savehyperref{#1}{Algorithm~\ref*{#1}}}
\newrefformat{rem}{\savehyperref{#1}{Remark~\ref*{#1}}}
\newrefformat{conj}{\savehyperref{#1}{Conjecture~\ref*{#1}}}
\newrefformat{prop}{\savehyperref{#1}{Proposition~\ref*{#1}}}
\newrefformat{proto}{\savehyperref{#1}{Protocol~\ref*{#1}}}
\newrefformat{prob}{\savehyperref{#1}{Problem~\ref*{#1}}}
\newrefformat{claim}{\savehyperref{#1}{Claim~\ref*{#1}}}

\def\A{{\mathcal A}}
\def\B{{\mathcal B}}

\def\E{{\mathcal E}}

\def\K{{\mathcal K}}
\def\L{{\mathcal L}}

\def\O{{\mathcal O}}

\def\R{{\mathcal R}}
\def\S{{\mathcal S}}

\def\U{{\mathcal U}}

\def\bbE{{\mathbb E}}

\renewcommand{\epsilon}{\varepsilon}
\newcommand{\tildeO}[1]{\widetilde{\mathcal{O}}\left(#1\right)}

%% file: abstract.tex
\begin{abstract}
We study the task of $(\epsilon, \delta)$-differentially private online convex optimization (OCO). In the online setting, the release of each distinct decision or iterate carries with it the potential for privacy loss. This problem has a long history of research starting with \cite{JKT12} and the best known results for the regime of $\epsilon$ not being very small are presented in \cite{pmlr-v195-agarwal23d}. In this paper we improve upon the results of \cite{pmlr-v195-agarwal23d} in terms of the dimension factors as well as removing the requirement of smoothness. Our results are now the best known rates for DP-OCO in this regime. 

Our algorithms builds upon the work of \citep{asi2023private} which introduced the idea of explicitly limiting the number of switches via rejection sampling. The main innovation in our algorithm is the use of sampling from a strongly log-concave density which allows us to trade-off the dimension factors better leading to improved results.
\end{abstract}


%% file: intro.tex
\section{Introduction}
\label{sec:intro}
In online convex optimization (OCO), in each round $t = 1, 2, \ldots, T$, a learner is required to  choose a point $x_t$ in a compact convex set $\K \in \reals^d$, and is provided an adversarially chosen Lipschitz convex loss function $l_t: \K \to \reals$ in response. The learner suffers loss $l_t(x_t)$ in round $t$. The learner's goal is to minimize her \emph{regret} defined as
$\textstyle\sum_{t=1}^T l_t(x_t) - \min_{x \in \K} \textstyle\sum_{t=1}^T l_t(x)$. We assume that the adversary chooses the loss functions \emph{obliviously}, i.e., independently of the points $x_t$. When the points $x_t$ are chosen randomly, the corresponding performance metric is the \emph{expected} regret.

\paragraph{Differentially Private OCO (DP-OCO).} The goal in DP-OCO is to design an online learning algorithm for this problem that guarantees that if one of the loss functions $l_t$ in an arbitrary round $t$ were changed to another function $l'_t$, then the {\em entire} output sequence of the algorithm doesn't change much in a certain precise manner depending on privacy parameters $(\epsilon, \delta)$ that we formalize later. DP-OCO has been studied for over a decade~\citep{JKT12,st13,AS17,kairouz2021practical,asi2023private,pmlr-v195-agarwal23d}. \citet{kairouz2021practical} established an upper bound for the regret which was $\tildeO{\frac{d^{1/4}\sqrt{T}}{\sqrt{\epsilon }}}$\footnote{$\tildeO{\cdot}$ hides polylog factors in $1/\delta$ and $T$.}. This was improved in a series of works ~\citep{asi2023private, pmlr-v195-agarwal23d}, for moderate ranges of $\varepsilon$, with \citet{pmlr-v195-agarwal23d} providing the best known bound of $\tildeO{\sqrt{T} +\frac{d T^{1/3}}{\epsilon}}$. There were two shortcomings of the result in \citet{pmlr-v195-agarwal23d}. Firstly they required the assumption that the convex functions are smooth. Secondly, \citet{pmlr-v195-agarwal23d} showed an improved bound of $\tildeO{\sqrt{T} +\frac{\sqrt{d} T^{1/3}}{\epsilon}}$ (improving the second term by a factor of $\sqrt{d}$) but only for the class of GLM functions. Obtaining the above bound for the general class of Lipschitz convex functions was left open. In this paper, we resolve this open problem. In particular we provide a DP-OCO algorithm for convex Lipschitz losses with $\tildeO{\sqrt{T}+\frac{\sqrt{d} T^{1/3}}{\epsilon }}$\footnote{There is a small regime of $\epsilon \in [T^{-1/6},d^{2/3}T^{-1/6}]$ where we get an additional term of $\frac{\sqrt{d}T^{3/8}}{\epsilon^{3/4}}$. Such a term also arises for \cite{pmlr-v195-agarwal23d} in the GLM case. See Table \ref{tab:dpoco} for detailed comparisons and \pref{thm:dpoco} for the precise bound on regret.} regret. This is now the best known result for DP-OCO. We provide a detailed comparison of our regret vs the previously best known algorithms in Table \ref{tab:dpoco}. 


\paragraph{Lazy OCO.} Lazy OCO is the problem of developing OCO algorithms with a limit on the number of switches between the points chosen by the learner. This setting is motivated by real-world applications where changes in the learner's decision are costly. For example, this cost manifests as the need for verifying the safety of the newly proposed controllers in robotics, transaction costs associated with rebalancing portfolios in portfolio optimization, and as the burden of reimplementation in public or organizational policy decisions.
Online learning with limited switching has been extensively studied in the context of prediction with expert advice \citep{MerhavOSW02,kalai2005efficient,GeulenVW10,AltschulerT21} and OCO \citep{anava2015online,sherman2021lazy}. For the OCO problem the best results known so far were provided in \citep{pmlr-v195-agarwal23d, sherman2021lazyx} who showed a regret bound of $\tilde{O}(\sqrt{T} + \frac{dT}{S})$ while switching at most $S$ times in expectation. In this setting again the above works assumed smoothness and in particular \cite{pmlr-v195-agarwal23d} showed an improved $\tilde{O}(\sqrt{T} + \frac{\sqrt{d}T}{S})$ under the additional assumption that the loss functions are GLMs. In this paper we improve these results by establishing an OCO algorithm that has regret at most $\tilde{O}(\sqrt{T} + \frac{\sqrt{d}T}{S})$ for Lischitz convex losses (without requiring smoothness or that the losses are GLMs).

\begin{table}[t]
    \centering
    \renewcommand{\arraystretch}{1.6}
\begin{tabular}{|c|c|c|c|c|}\hline
\multirow{2}{*}{$\epsilon$}  &\multicolumn{2}{c|}{Previous Best}  & Our Algorithm \\\cline{2-3}
 &Require Smoothness & No Smoothness & (No Smoothness) \\\hline
$\epsilon \geq dT^{-1/6}$ & $\sqrt{T}$ \citep{pmlr-v195-agarwal23d} & \multirow{3}{*}{\makecell{$\sqrt{dT}$ \\ \citep{asi2023private}}} & $\sqrt{T}$  \\\cline{1-2}\cline{4-4}
$ \epsilon \in [d^{2/3}T^{-1/6}, dT^{-1/6}]$ & \multirow{4}{*}{\makecell{$d \cdot T^{1/3}\cdot \epsilon^{-1}$ \\ \citep{pmlr-v195-agarwal23d}}} &  &\color{red}{$\sqrt{T}$}  \\\cline{1-1}\cline{4-4}
$\epsilon \in [\sqrt{d}T^{-1/6}, d^{2/3}T^{-1/6}]$  & &  &\multirow{2}{*}{\color{red}{$\sqrt{d} \cdot T^{3/8} \cdot \epsilon^{-3/4}$}}  \\\cline{1-1}\cline{3-3}
$\epsilon \in [T^{-1/6},\sqrt{d}T^{-1/6}]$  & & \multirow{2}{*}{\makecell{$d \cdot T^{1/3}\cdot \epsilon^{-1}$ \\ \citep{asi2023private}}} &   \\\cline{1-1}\cline{4-4}
$\epsilon \in [d^{3/2}T^{-1/3},T^{-1/6}]$ &  & & \multirow{2}{*}{\color{red}{$\sqrt{d} \cdot T^{1/3}\cdot \epsilon^{-1}$}} \\ \cline{1-3}
$\epsilon \in [dT^{-1/3},d^{3/2}T^{-1/3}]$ & \multicolumn{2}{c|}{$d^{1/4} \cdot T^{1/2} \cdot \epsilon^{-1/2}$  \citep{kairouz2021practical}} & \\ \hline
$\epsilon \leq dT^{-1/3}$ & \multicolumn{2}{c|}{\makecell{$d^{1/4} \cdot T^{1/2} \cdot \epsilon^{-1/2}$  \citep{kairouz2021practical}\\ (Current Best)}} & $\sqrt{d} \cdot T^{1/3}\cdot \epsilon^{-1}$\\\hline
\end{tabular}
    \caption{Landscape of the best known results for DP-OCO across different regimes. We highlight our results with the color {\color{red}red} in all the regimes where we are \textit{strictly better} (in terms of $d$) than the best known result. Notice that our algorithm strictly improves the best known results \textit{without assuming smoothness} for all $\epsilon \geq dT^{1/3}$. While we focus on factors of $d$, for the asymptotics we assume $T \gg d$.}
    \label{tab:dpoco}
\end{table}

\begin{table}[h!]
    \centering
    \renewcommand{\arraystretch}{1.3}
\begin{tabular}{|c|c|c|c|}\hline
\multirow{4}{*}{$S$} & Previous Best & \multirow{4}{*}{Our Algorithm} \\
& \cite{pmlr-v195-agarwal23d} & \\
& \cite{sherman2021lazyx} & \\
& (Assumes Smoothness) & \\ \hline
$S \geq d\sqrt{T}$ &$\sqrt{T}$ &$\sqrt{T}$ \\\hline
$\sqrt{dT} \leq S \leq d\sqrt{T}$ & $\frac{dT}{S}$  &$\color{red}{\sqrt{T}}$ \\
\hline
$S \leq \sqrt{dT}$ & $\frac{dT}{S}$  &\color{red}{$\frac{\sqrt{d} \cdot T}{S}$} \\ \hline
\end{tabular}
    \caption{Comparison between our results and the known best results previously for Lazy OCO \citep{pmlr-v195-agarwal23d, sherman2021lazyx} in different regimes for the switching budget $S$. We highlight our results with the color {\color{red}red} in all the regimes where we are \textit{strictly better} (in terms of $d$) than the best known result. While we focus on factors of $d$, for the asymptotics we assume $T \gg d$.}
    \label{tab:lazyoco}
\end{table}

%% file: problem.tex
\section{Preliminaries}
\newcommand{\TV}{\mathrm{TV}}

\paragraph{Notation.} We use $\|\cdot\|$ to denote the standard $\ell_2$ norm in $\reals^d$. For distributions $p$ and $q$ on the same outcome space, we use $\|p - q\|_\TV$ to denote their total variation distance. For a distribution $\mu$ on $\reals^d$, we use $\mu(A)$ to denote the measure of a measurable set $A \subseteq \reals^d$. With some abuse of notation, we also $\mu(x)$ to denote the density of $\mu$ at $x \in \reals^d$, if it exists.

\paragraph{Problem Setting.} We are given a convex compact set $\K \in \reals^d$ with diameter $D$ (i.e. $D = \max_{x,y\in \K} \|x-y\|$). In OCO, at the start of each round $t\in [T]$, the learner $\mathcal{A}$ chooses a point $x_t\in\K$ from some compact convex decision set $\K\subset \reals^d$, and upon making this choice it observes the loss function $l_t:\K\to \reals$, and suffers a loss of $l_t(x_t)$. For any $t$-indexed sequence of objects, e.g. the loss function $l_t$, let $l_{1:T} = (l_1, \dots l_T)$ be the concatenated sequence. We restrict our attention to the case of {\em oblivious adversaries} in that we assume the loss function sequence $l_{1:T}$ is chosen independently of the iterates $x_{t}$ picked by the learner.\footnote{As remarked in \cite{asi2023private}, and as is true for most of the literature on private OCO, our privacy bounds hold in the absence of this assumption -- obliviousness -- due to the use of adaptive strong composition. The utility or regret bounds are strongly reliant on this assumption on loss functions, however.} Recall that a function $l:\K\to\reals$ is said to be $G$-Lipschitz if $|l(x)-l(y)|\leq G\|x-y\|$ for any pair $x,y\in \K$. As for the domain $\K$, we assume that (a) it is full-dimensional and (b) $0 \in \K$. 

\begin{assumption}
	The loss functions $l_{1:T}\in \mathcal{L}^T$ are chosen {\em obliviously} from the class $\mathcal{L}$ of $G$-Lipschitz twice-differentiable convex functions.
\end{assumption}

The possibly random learner's performance through such mode of interaction as outlined above may be assessed via the regret it incurs; this, as defined below, measures the expected excess aggregate loss the learner is subject to in comparison to the best fixed point in $\K$ determined with the benefit of hindsight.
$$ \R_T(\mathcal{A}, l_{1:T}) \defeq \mathop{\mathbb{E}}_{\mathcal{A}}\left[ \sum_{t=1}^T l_t(x_t) - \min_{x^*\in \K}\sum_{t=1}^T l_t(x^*) \right]$$
Later on, since we do not make any distributional assumptions on the loss sequence, the primary quantity of interest will be the {\em worst-case} regret, i.e. $\R_T(\mathcal{A}) \defeq \max_{l_{1:T}\in \mathcal{L}^T} \R_T(\mathcal{A}, l_{1:T})$. 

Another characteristic of the learner that is relevant to the discussion below is the number of discrete switches the learner makes. To this end, we define the number of switches the learner makes as
$$ \S_T(\mathcal{A}, l_{1:T}) \defeq \bbE_{\A}\left[\sum_{t=2}^T \mathbb{I}_{x_t \neq x_{t-1}}\right].$$
For brevity, henceforth we will simply use $\R_T$ and $\S_T$ to refer to $\R_T(\mathcal{A}, l_{1:T})$ and $\S_T(\mathcal{A}, l_{1:T})$ respectively.


Finally, an online learning algorithm $\mathcal{A}$ is said to $(\varepsilon,\delta)$-differentially private if for any loss function sequence pair $l_{1:T}, l'_{1:T}\in \L^T$ such that $l_t = l'_t$ for all but possibly one $t \in [T]$, we have for any Lebesgue measurable $O\subset \K^T$ that
$$\Pr_{\mathcal{A}}(x_{1:T}\in O|l_{1:T}) \leq e^\varepsilon \Pr_{\mathcal{A}}(x_{1:T}\in O|l'_{1:T}) + \delta.$$

%% file: result.tex
\section{Preliminary results for Gibbs measures}

In this paper we consider a class of Gibbs distributions over the set $\K$. Given any function $f:\K \in \reals$, a temperature constant $\beta \geq 0$ and a regularization parameter $\lambda \geq 0$ we define $\mu(f, \beta, \lambda): \K \rightarrow \reals_+$ to be a measure function defined as 
\begin{equation}
    \label{eqn:mu-def}
    \mu(f, \beta, \lambda)(x) = \exp \left( - \beta \cdot \left(f(x) + \frac{\lambda}{2}\|x\|^2 \right)\right).
\end{equation}
We further define $Z(f, \beta, \lambda)$ to be normalization constant of the above function defined as
\begin{equation}
    \label{eqn:z-def}
Z(f, \beta, \lambda) = \int_{x \in \K} \exp \left( - \beta \cdot \left(f(x) + \frac{\lambda}{2}\|x\|^2 \right)\right) dx.
\end{equation}
Using the above we can define a probability density $\barmu(f, \beta, \lambda)(x)$ over $\K$ as follows 
\begin{equation}
    \label{eqn:mu-bar-def}
\barmu(f, \beta, \lambda)(x) \defeq \frac{\mu(f, \beta, \lambda)(x)}{Z(f, \beta, \lambda)}.
\end{equation}
We will interchangeably use the notation $\barmu$ for the probability density function as well as the distribution itself. We will suppress $\beta, \lambda$ from the above definitions when they will be clear from the context. In the following we collect some useful definitions and results pertaining to concentration of measure resulting from the Log-Sobolev Inequality.
\begin{definition}\label{def:lsi}
    A distribution $P$ satisfies the Log-Sobolev Inequality (LSI) with constant $c$ if for all smooth functions $g: \reals^{d} \rightarrow \reals$ with $\bbE_{x \sim P}[g(x)^2] < \infty$:
    \[ \bbE_{x \sim P}[g(x)^2\log(g(x)^2)] - \bbE_{x \sim P}[g(x)^2] \bbE_{x \sim P}[\log(g(x)^2)] \leq \frac{2}{c}\bbE_{x \sim P}[\|\nabla g(x)\|^2]\]
\end{definition}

\begin{lemma}[Proposition 3 and Corollaire 2 in \cite{bakry2006diffusions}]
\label{lem:psi-sc}
    Given a $\Lambda$-strongly convex function $l$, let $Q$ be the distribution supported over $\K$ with density $\mu(x)$ proportional to $\exp\left( -\beta \cdot l(x) \right)$. Then Q satisfies LSI (\pref{def:lsi}) with constant $c = \beta \Lambda$. 
\end{lemma}

\begin{lemma}[Concentration of Measure (follows from Herbst's argument presented in Section 2.3 \cite{ledoux1999concentration})]
\label{lem:psi-sc}
    Let $F$ be a $L$-Lipschitz function and let $Q$ be a distribution satisfying LSI with a constant $c$ then 
    \[ \Pr_{X \sim Q}(|F(X) - \bbE[F(X)]| \geq r) \leq 2 \exp\left(-\frac{c \cdot r^2}{2L^2}\right)\]
\end{lemma}
The following definition defines a notion of closeness for two Gibbs-measures:
\begin{definition} \label{def:phidelta} Two Gibbs distributions $\barmu$, $\barmu'$ on $\K$ are said to be $(\Phi, \delta)$-close if
\[\Pr_{X \sim \barmu}\left[\frac{1}{\Phi} \leq \frac{\barmu(X)}{\barmu'(X)} \leq \Phi\right] \geq 1 - \delta \quad \text{ and } \quad \Pr_{X \sim \barmu'}\left[\frac{1}{\Phi} \leq \frac{\barmu(X)}{\barmu'(X)} \leq \Phi\right] \geq 1 - \delta.\]
\end{definition}

One of the core components of our analysis is to show that the Gibbs-measures are \textit{smooth} under changes of the underlying functions.
\begin{replemma}{lem:densityratio}[Density ratio]
Let $l, l': \K \to \reals$ be convex twice-differentiable functions such that $l - l'$ is $G$-Lipschitz. Further let $\beta, \lambda \geq 0$ be parameters and define the Gibbs-distributions $\barmu = \barmu(l, \beta, \lambda)$ and $\barmu'=\barmu(l', \beta, \lambda)$ (as defined in \eqref{eqn:mu-def}). Then for any $\delta \in (0, 1]$, we have that $\barmu$ and $\barmu'$ are $(\Phi, \delta)$ close where 
\[ \Phi = \exp \left( \frac{2\beta G^2}{\lambda} + \sqrt{\frac{8 \beta G^2 \log(2/\delta)}{\lambda}} \right)\]
\end{replemma}



The proof of the lemma crucially uses uses the following bound on the Wasserstein-distance of the Gibbs-distributions and other machinery developed by \cite{pmlr-v195-ganesh23a}.

\begin{replemma}{lem:waserstein-dist}[Wasserstein Distance]
Let $l, l': \K \to \reals$ be convex twice-differentiable functions such that $l - l'$ is $G$-Lipschitz. Further let $\beta, \lambda \geq 0$ be parameters and define the Gibbs-distributions $\barmu = \barmu(l, \beta, \lambda)$ and $\barmu'=\mu(l', \beta, \lambda)$ (as defined in \eqref{eqn:mu-bar-def}). 
Then we have that $\infty$-Wasserstein distance between $\barmu$ and $\barmu'$ over the $\ell_2$ metric is bounded as 
\[ W_{\infty}(\barmu, \barmu') \leq \frac{G}{\lambda}.\]
\end{replemma}

\section{Algorithm and main result}

\newcommand{\xstar}{x^\star}


\begin{algorithm2e}[t]
\caption{Private Continuous Online Multiplicative Weights (P-OCMW)} \label{alg:ctrl}
    \textbf{Inputs:} A temperature parameter $\beta$, a regularization parameter $\lambda > 0$, switching rate parameter $p \in [0, 1]$, switching budget $B \geq 0$, a scale parameter $\Phi > 0$. \\
    Choose $x_1 \sim $. \\
    \For{$t=1$ to $T$}{
      Play $x_t\in \K$.\\
        Observe $l_t:\K \to \reals$ and suffer a loss of $l_t(x_t)$.\\
        Define the measure function $\mu_{t+1}(x) \defeq \mu(f,\beta,\lambda)(x) \defeq \exp\left( -\beta \left( \sum_{\tau=1}^t l_{\tau}(x)  + \lambda \frac{\|x\|^2}{2} \right) \right)$ \\
        Accordingly denote $\barmu_{t+1}(x)$ the probability density resulting from the measure $\mu_{t+1}$. (cf. \eqref{eqn:mu-bar-def}) \\
        Sample $S_t \sim \Ber\left(\min\left\{1,\max\left\{\frac{1}{\Phi^2}, \frac{\barmu_{t+1}(x_t)}{\Phi \cdot \barmu_{t} (x_t)}\right\}\right\}\right)$ and $S_t' \sim \Ber(1-p)$.\\
        \If{$b_t < B$ and  ($S_t'=0$ or $S_t=0$)}{
            Update $b_{t+1}=b_t+1$ and draw an independent sample $x_{t+1} \sim \barmu_{t+1}$.
        }\Else{
          Set $b_{t+1}=b_{t}$ and $x_{t+1}=x_t$. 
      }
    }
\end{algorithm2e}

Our proposed algorithm Private Continuous Online Multiplicative Weights (P-COMW) (\pref{alg:ctrl}) is a small modification of the Private Shrinking Dartboard algorithm proposed by \cite{asi2023private} (also see \cite{pmlr-v195-agarwal23d}). At a high level at every step the algorithm ensures that at every iteration it samples $x_t$ marginally from the distribution $\barmu_t$ over $\K$ corresponding to the measure function $\mu_t(x)$ defined as 
\[ \mu_t(x) = \mu\left(\sum_{\tau = 1}^t l_{\tau}, \beta, \lambda\right) = \exp\left( -\beta \left( \sum_{\tau=1}^{t-1} l_{\tau}(x) + \lambda \frac{\|x\|^2}{2}  \right) \right).\]
The distribution $\barmu_t$ is the same distribution as Online Continuous Multiplicative Weights (as used in \cite{asi2023private}) with an added strong-convexity term governed by $\lambda$. This additional strong-convexity term is key to the improvements provided in this paper as it provides a better trade-off between switching and regret. 

The above scheme was first analyzed by \cite{pmlr-v178-gopi22a} and was recently shown to be able to obtain optimal results in the case of stochastic convex optimization \cite{pmlr-v195-ganesh23a}. In the online case however a direct application of the above scheme can a lot of private information since the algorithm can potentially alter its decisions in each round. To guard against this, as in the work of \cite{asi2023private, pmlr-v195-agarwal23d}, we use a rejection sampling procedure which draws inspiration from \cite{GeulenVW10}. Specifically, for any $t$, the point $x_{t+1}$ is chosen to be equal to $x_{t}$ with probability $\frac{\barmu_{t+1}(x_t)}{\Phi \barmu_t(x_t)}$, where $\Phi$ is a scaling factor. With the remaining probability, we sample $x_{t+1}$ independently from $\barmu_{t+1}$ (we call this a ``switch''). This rejection sampling technique ensures that the distribution of $x_{t+1}$ remains very close to $\barmu_{t+1}$. We rescale the density ratio $\frac{\barmu_{t+1}(x_t)}{\Phi \barmu_t(x_t)}$ appropriately to make sure it is at most unit sized with high probability.

\newcommand{\ptil}{\tilde{p}}

We now turn to the regret analysis for \pref{alg:ctrl}. We have the following regret bound for \pref{alg:ctrl}. 

\begin{reptheorem}{thm:regret}[Regret bound for P-COMW]
In \pref{alg:ctrl}, fix any $\beta, \lambda>0$, any $\delta \in [0,1/2]$, any $p \in [0,1]$, and choose $\Phi$ such that for all $t$ the distributions $\barmu_t, \barmu_{t+1}$ are $(\Phi, \delta)$-close. 
For any sequence of obliviously chosen $G$-Lipschitz, convex loss functions $l_{1:T}$, the following hold:
\begin{itemize}
      \item If $B=\infty$,
	$$ \R_T \leq \frac{\lambda D^2}{2} + \frac{G^2T}{\lambda} + \frac{d \log(T)}{\beta} + GD + 6GD\delta T^2 .$$
      \item Let $\ptil = p + 1 - \Phi^{-2}$. If $B=3\ptil T$,
	$$\R_T \leq \frac{\lambda D^2}{2} + \frac{G^2T}{\lambda} + \frac{d \log(T)}{\beta} + 2GDT(e^{-\ptil T} + 3\delta T) + GD.$$
 \end{itemize}
 \end{reptheorem}

The following lemma (originally proved in \cite{pmlr-v195-agarwal23d}), gives a bound on the number of switches made by the \pref{alg:ctrl} and immediately follows by observing that the probability of switching in any round is at most $\ptil$ via a simple Chernoff bound. For completeness we provide a proof in \pref{app:analysis}.
\begin{replemma}{lem:glm-hp}[Switching bound]
For any $p \in [0,1]$ and any $\Phi \geq 0$, setting $\ptil = p + 1 - \Phi^{-2}$, we have that the number of switches is bounded in the following manner,
\[ \bbE\left[\S_T\right] \leq \ptil T, \qquad \Pr\left[\S_T \geq 3\ptil T\right] \leq e^{-\ptil T}.\]
\end{replemma}

Finally, we turn to the privacy guarantee for \pref{alg:ctrl}.
\begin{reptheorem}{thm:dp}[Privacy]
Given $\beta, \lambda > 0$ and $\delta \in (0,1/2]$, for any $T \geq 12\log(1/\delta)$, let $\delta' = \frac{\delta T^{-2}}{60}$, $G' = 3G$. Suppose there exists $\Phi' > 0$ such that for all convex functions $l, l'$ where $l-l'$ is $G'$-Lipschitz, we have that, the distributions $\barmu(l, \beta, \lambda)$ and $\barmu(l', \beta, \lambda)$ respectively are $(\Phi',\delta')$-close. Then for any sequence of $G$-Lipschitz convex functions, \pref{alg:ctrl} when run with $\Phi = {\Phi'}^2$, $p = \max\left(T^{-1/3}, \left( \frac{G^4 \beta^2}{\lambda^2 \cdot \log^2(\Phi)} \right)^{1/3}\right)$, $\ptil=p+1-\Phi^{-2}$ and $B=3\ptil T$ is $(\varepsilon, \delta + 3Te^{-(1-\Phi^{-2})T})$-differentially private where
\[\varepsilon = 3\varepsilon'/2  + \sqrt{6 \varepsilon'}\sqrt{\log(2/\delta)},\]
with
\[\varepsilon' =  7T^{2/3}\log^2({\Phi}) + 12\log^3({\Phi})T + 11 \left(\frac{G^4 \beta^2}{\lambda^2}\right)^{1/3}\log^{4/3}(\Phi) \cdot T.\]
\end{reptheorem}

\subsection{Bounds for Lipschitz loss functions}

In order to apply the above results for OCO with convex $G$-Lipschitz loss functions, all we need to do is compute $\Phi$. This bound was established by \pref{lem:densityratio}. Using \pref{lem:densityratio} and combining \pref{thm:dp} and \pref{thm:regret}, we get the following result via straightforward calculations:
\begin{theorem}[DP OCO]\label{thm:dpoco}
For any given $\epsilon \leq 1, \delta \in (0,1/2]$ and any $T \geq 12\log(1/\delta)$, set \[\lambda = \frac{G}{D} \max \left( \frac{1}{2\sqrt{T}}, \frac{10^3 T^{1/3} \sqrt{d} \log(T/\delta)}{ \epsilon}, \frac{10^3 T^{3/8} \sqrt{d} \log(T/\delta)}{ \epsilon^{3/4}} \right)\]
\[\beta = \frac{\lambda}{10^5 \cdot G^2 \log^2(T/\delta)}\min \left( \frac{\epsilon^{2}}{T^{2/3}}, \frac{\epsilon^{3/2}}{T^{3/4}} \right)\]
and other parameters as in \pref{thm:dp}. Then we get that \pref{alg:ctrl} is $(\epsilon,\delta)$ differentially private and additionally satisfies
$$ \R_T \leq \widetilde{\O}\left(GD\sqrt{T} + GD\log^2(T/\delta)\cdot \sqrt{d} \left(\frac{T^{1/3}}{\epsilon} + \frac{T^{3/8}}{\epsilon^{3/4}} \right)\right).$$
\end{theorem}
Similarly, for Lazy OCO, using \pref{thm:regret}, \pref{lem:glm-hp} and \pref{lem:densityratio}, we get the following result:

\begin{theorem}[Lazy OCO] \label{thm:LazyOCO-general}
For any $T \geq 3$ and any given bound $S \leq T$ on the number of switches, set $\delta = 2/T^2$, $\lambda = \max\left\{ \frac{G \sqrt{2T}}{D}, \frac{\sqrt{512d}G \log(T)}{D} \cdot \frac{T}{S} \right\}$, $\beta = \frac{\lambda}{256G^2 \log(T)} \cdot \frac{S^2}{T^2}$, 
$\Phi = \exp \left( \frac{2\beta G^2}{\lambda} + \sqrt{\frac{8 \beta G^2 \log(2/\delta)}{\lambda}} \right)$, $p = 0$ and $B=\infty$ in \pref{alg:ctrl}. Then for any sequence of obliviously chosen $G$-Lipschitz convex loss functions $l_{1:T}$, \pref{alg:ctrl} satisfies the following:
 $$\R_T \leq GD \sqrt{2T} + 16GD \log(T) \cdot \frac{\sqrt{d} \cdot T}{S} + 13GD \text{ and } \bbE[\S_T] \leq S.$$
\end{theorem}
\begin{proof}
We begin by first bounding the number of switches using Lemma \ref{lem:glm-hp}. We get that 
\[\bbE[\S_T] \leq \ptil T \leq (1 - \Phi^{-2})T \leq 2\log(\Phi)T \leq 2T \left( \underbrace{\frac{2\beta G^2}{\lambda}}_{=\frac{S^2}{128T^2\log(T)} \leq \frac{S}{128T}} + \underbrace{\sqrt{\frac{8\beta G^2\log(2/\delta)}{ \lambda }}}_{\leq \frac{S}{4T}} \right) \leq S\]
To bound the regret note that \pref{lem:densityratio} implies that the distributions $\mu_t, \mu_{t+1}$ are $(\Phi, \delta)$-close and therefore \pref{thm:regret} implies 
\begin{align*}
    \R_T &\leq \frac{\lambda D^2}{2} + \frac{G^2T}{\lambda} + \frac{d \log(T)}{\beta} + GD + 6GD\delta T^2 \\
    &= \frac{\lambda D^2}{2} + \frac{G^2T}{\lambda} + \frac{256 \cdot d \cdot G^2 \log^2(T)}{\lambda} \cdot \frac{T^2}{S^2} + 13GD\\
    &\leq GD \sqrt{2T} + 16GD \log(T) \cdot \frac{\sqrt{d} \cdot T}{S} + 13GD.
\end{align*}
\end{proof}

%% file: prelim.tex
\section{Useful Results}
In this section, we recall some standard results in differential privacy and online learning. The first of these standard results is the adaptive strong composition lemma for differentially private mechanisms.

\begin{lemma}[e.g. \cite{whitehouse2022fully}]\label{lem:comp}
Let $\mathcal{A}_t:\L^{t-1}\times \K^{t-1}\to \K$ be a $t$-indexed family of $(\varepsilon_t,\delta_t)$-differentially private algorithms, i.e. for every $t$, for any pair of sequences of loss functions $l_{1:t-1},l'_{1:t-1}\in \L^{t-1}$ differing in at most one index in $[t-1]$, and any $x_{1:t-1}\in \K^{t-1}$, it holds that 
$$P_{\mathcal{A}_t}(x_{t}|l_{1:t-1}, x_{1:t-1}) \leq e^\varepsilon P_{\mathcal{A}_t}(x_{t}|l'_{1:t-1},x_{1:t-1}) + \delta.$$
Define a new $t$-indexed family $\mathcal{B}_t:\L^{t-1}\to \K^t$ recursively starting with $\mathcal{B}_1 = \mathcal{A}_1$ as
$$ \mathcal{B}_t(l_{1:t-1}) = \mathcal{B}_{t-1}(l_{1:t-2}) \circ \mathcal{A}_t(l_{1:t-1}, \mathcal{B}_{t-1}(l_{1:t-2})).$$
Then for any $\delta''>0$, $\mathcal{B}_T$ is $(\varepsilon', \delta')$-differentially private, where 
$$
\varepsilon' = \frac{3}{2}\sum_{t=1}^T \varepsilon_t^2 + \sqrt{6\sum_{t=1}^T \varepsilon_t^2 \log \frac{1}{\delta''}}, \qquad \delta' = \delta'' + \sum_{t=1}^T \delta_t.
$$
\end{lemma}

Next, we state the follow-the-leader be-the-leader lemma that is helpful in bounding the regret of an online learner as a sum of stability-related and regularization-related terms.

\begin{lemma}[FTL-BTL \cite{hazan2016introduction}]\label{lem:ftlbtl}
For any loss function sequence $l_{0:T}$ over any set $\B$, define $$y_t = \argmin_{x\in \B}\left\{\sum_{i=0}^{t-1} l_i(x)\right\}.$$ Then, for any $x\in \B$, we have
$$ \sum_{t=0}^T l_t(y_{t+1}) \leq \sum_{t=0}^T l_t(x) .$$
\end{lemma}


%% file: density-smoothness-proofs.tex
\section{Proofs of smoothness of Gibbs measures}
In this section we prove the Lemmas concerning the smoothness of Gibbs measures, i.e. Lemmas \ref{lem:waserstein-dist} and  \ref{lem:densityratio}. We begin by restating and proving \pref{lem:waserstein-dist}. 
\repeatlemma{lem:waserstein-dist}
\begin{proof}[Proof of \pref{lem:waserstein-dist}]
    By definition $W_{\infty}(\barmu, \barmu') = \inf_{\gamma \in \Gamma(\barmu,\barmu')} \sup_{(X, X') \sim \gamma} \|X - X'\|$, where the notation $\sup_{(X, X') \sim \gamma}$ is shorthand for all $(X, X')$ in the support of $\gamma$. To bound $W_{\infty}$ we consider the following coupling between $\barmu, \barmu'$. Define the functions $L(x)=l(x) + \frac{\lambda}{2}\|x\|^2$, $L'(x)=l'(x) + \frac{\lambda}{2}\|x\|^2$ and consider the following "Projected" Langevin diffusions given by the following SDEs (see \cite{pmlr-v195-ganesh23a} for details):
    \[ d X_{t+1} = -\beta \nabla L(X_t) + \sqrt{2}dW_t - \nu_t \zeta(d_t) \]
    \[ d X'_{t+1} = -\beta \nabla L'(X_t) + \sqrt{2}dW_t - \nu_t' \zeta'(d_t) \]
    where $\zeta$ and $\zeta'$ are measures supported on $\{t: X_t \in \partial K\}$ and $\{t: X_t' \in \partial K\}$ respectively, and $\nu_t$ and $\nu_t'$ are outer unit normal vectors at $X_t$ and $X_t'$ respectively. It is known that $\lim_{t \rightarrow \infty} X_{t}$ converges in distribution to $\barmu$ and similarly $\lim_{t \rightarrow \infty} X'_{t}$ converges in distribution to $\barmu'$. Our desired coupling $\gamma$ is defined by sampling a  Brownian motion sequence $\{W_t\}_{t=1}^{\infty}$ and the output sample is set to $\lim_{t \rightarrow \infty} X_{t}$ and $\lim_{t \rightarrow \infty} X'_{t}$ with the same $\{W_t\}_{t=1}^{\infty}$ sequence. For a fixed Brownian motion sequence $\{W_t\}_{t=1}^{\infty}$, we get the following calculations (by defining $\Delta_t = \|X_t - X'_t\|$):
    \begin{align*}
        \frac{1}{2}\frac{d \Delta_t^2}{dt} = \frac{1}{2}\frac{d \|X_t - X'_t\|^2}{dt} &= \langle \frac{d X_t}{dt} - \frac{d X'_t}{dt}, X_t - X'_t \rangle \\
        &= -\beta \langle \nabla l(X_t) - \nabla l'(X'_t), X_t - X'_t \rangle - \langle \nu_t, X_t - X_t'\rangle \frac{\zeta(d_t)}{d_t} + \langle \nu_t', X_t - X_t'\rangle \frac{\zeta'(d_t)}{d_t}\\
        &\leq -\beta \langle \nabla l(X_t) - \nabla l'(X'_t), X_t - X'_t \rangle \\
        & (\because \langle \nu_t, X_t' - X_t\rangle \leq 0 \text{ and } \langle \nu_t', X_t - X_t'\rangle \leq 0 \text{ since } K \text{ is convex}) \\
        &= -\beta \langle \nabla l(X_t) - \nabla l(X'_t), X_t - X'_t \rangle + \beta \langle \nabla l'(X'_t) - \nabla l(X'_t), X_t - X'_t \rangle \\
        &\leq \beta \left(-\lambda\|X_t - X'_t\|^2  + G\|X_t - X'_t\| \right) =  \beta \left(-\lambda\Delta_t^2  + G\Delta_t \right) \\
        &\leq \beta \left(-\lambda\Delta_t^2  + \frac{\lambda}{2} \Delta_t^2 + \frac{G^2}{2 \lambda}  \right) \\
        &\leq \beta \left(- \frac{\lambda}{2} \Delta_t^2 + \frac{G^2}{2 \lambda}  \right)
    \end{align*}
    
    Defining $F_t = \Delta_t^2 - \frac{G^2}{\lambda^2}$, the above implies that $\frac{dF_t}{dt} \leq - \beta \lambda F_t$ which implies, via Gr\"onwall's inequality, that $F_t \leq F_0\exp(-\beta \lambda t)$. Therefore we have that $\lim_{t \rightarrow \infty} F_t \rightarrow 0$ which implies that $\lim_{t \rightarrow \infty} \Delta_t \rightarrow \frac{G}{\lambda}$. 

    Therefore we get that under the above coupling $\gamma$ we have that $\sup_{(x, y) \sim \gamma} \|x - y\| \leq \frac{G}{\lambda}$ which finishes the proof. 
\end{proof}


Using the above we restate and prove \pref{lem:densityratio} below. 
\repeatlemma{lem:densityratio}
\begin{proof}[Proof of \pref{lem:densityratio}]
We begin first by proving the direction \[\Pr_{X \sim \barmu}\left[\frac{1}{\Phi} \leq \frac{\barmu(X)}{\barmu'(X)} \leq \Phi\right] \geq 1 - \delta\]
and reverse direction follows easily by switching the roles of $\barmu, \barmu'$ through the analysis. To this end define the function $g(X) = \log\left(\frac{\barmu(X)}{\barmu'(X)}\right)$. Therefore we are required to show that 
$$\Pr_{X \sim \barmu}(|g(X)| > \log(\Phi)) \leq \delta.$$
We will show this by first bounding $\bbE_{X \sim \barmu}[g(X)]$ and then showing that it concentrates around its expectation. We first show that $g$ is a $2\beta G$-Lipschitz function. To see this consider the following 
\[ |g(X) - g(X')| = \big|\log\left(\frac{\barmu(X)}{\barmu(X')}\right) + \log\left(\frac{\barmu'(X')}{\barmu'(X)}\right)\big| = |-\beta \left( l(X) - l(X') + l'(X') - l'(X) \right)| \leq 2\beta G\|X- X'\|.\]

Using the proof of \pref{lem:waserstein-dist} we get that there is a coupling $\gamma$ between $\barmu, \barmu'$ such that $\sup_{(X, X') \sim \gamma} \|X - X'\| \leq \frac{G}{\Lambda}$, therefore sampling from the coupling and using the Lipschitzness of $g$, we get that 
\[\bbE_{(X, X') \sim \barmu}[|g(X) - g(X')|] \leq 2\beta G\cdot \bbE_{(X, X') \sim \barmu}[\|X - X'\|] \leq  2\beta G\cdot\frac{G}{\Lambda},\]
which implies that
\[ \bbE_{X \sim \barmu}[g(X)] \leq \bbE_{X' \sim \barmu'}[g(X')] + 2\beta G\cdot \frac{G}{\Lambda}\]
Now noticing that $\bbE_{X \sim \barmu'}[g(X)] = -\text{KL}(\barmu' \| \barmu) \leq 0$, we get that 
\[ \bbE_{X \sim \barmu}[g(X)] \leq  \frac{2\beta G^2}{\Lambda}.\]
Furthermore, note that  $\bbE_{X \sim \barmu}[g(X)] = \text{KL}(\barmu \| \barmu') \geq 0$.
Thus, we have
\[ 0 \leq \bbE_{X \sim \barmu}[g(X)] \leq  \frac{2\beta G^2}{\Lambda}.\]
Next we give a high probability bound on $g$. \pref{lem:psi-sc} implies that the distribution correspnding to $\barmu$ satisfies LSI (\pref{def:lsi}) with constant $\beta \Lambda$. Now by Proposition 2.3 in \cite{ledoux1999concentration}, plugging in the LSI constant and Lipschitzness bound for g, we have that
\[ \Pr_{X \sim \barmu} [| g(X) - \bbE[g(X)] | \geq r ] \leq 2\exp \left( - \frac{\Lambda r^2}{8\beta G^2 } \right)\]
 Thereby setting $r = \sqrt{\frac{8 \beta G^2 \log(2/\delta)}{\Lambda}}$ we have that
 $$\Pr_{X \sim \barmu}\left(|g(X)| > \frac{2\beta G^2}{\Lambda} + \sqrt{\frac{8 \beta G^2 \log(2/\delta)}{\Lambda}} \right) \leq \delta.$$

\end{proof}

%% file: general-analysis.tex
\section{Analysis of \pref{alg:ctrl}}
\label{app:analysis}

For notational convenience, define $\Pi: \reals \to [\frac{1}{\Phi^2}, 1]$ as $\Pi(x) = \min\{1, \max\{\frac{1}{\Phi^2}, x)\}\}$. Also define $\zeta_t \defeq \mathbb{I}(S_t = 0 \text{ or } S'_t = 0)$. We restate and prove \pref{lem:glm-hp} first:
\repeatlemma{lem:glm-hp}
\begin{proof}
Since $S_t \sim \text{Ber}\left(\Pi\left(\frac{\mu_{t+1}(x_t)}{\Phi\mu_t(x_t)}\right)\right)$, we have $\Pr[S_t = 0] \leq 1-\Phi^{-2}$. From the definition of $\zeta_t$, we have
\begin{align}
    \bbE[\zeta_t] = \Pr(S_t' = 0) + (1 - \Pr(S_t' = 0))\cdot\Pr(S_t = 0) \leq p + (1-p) \cdot (1-\Phi^{-2}) \leq \ptil.
\end{align}
Thus, the random variable $S_T = \sum_{t=1}^T\zeta_t$ is stochastically dominated by the sum of $T$ Bernoulli random variables with parameter $\ptil$. Hence, $\bbE[\S_T] \leq \ptil T$ and the Chernoff bound\footnote{The specific bound used is that for independent Bernoulli random variables $X_1, X_2, \ldots, X_T$, if $\mu = \bbE[\sum_{t=1}^T X_t]$, then for any $\delta > 0$, we have $\Pr[\sum_{t=1}^T X_t \geq (1+\delta)\mu] \leq e^{-\delta^2\mu/(2+\delta)}$.} implies
$$ \Pr\left[\S_T \geq 3\ptil T\right] \leq e^{-\ptil T}.$$
\end{proof}
The following key lemma (adapted from \cite{pmlr-v195-agarwal23d}) obtains bounds on the actual distribution that $x_t$ is sampled from in terms of $\barmu_t$:
\begin{lemma}[Distribution drift]\label{lem:samedist}
Given $\delta \in [0, \frac{1}{2}]$ and $\Phi \geq 1$, suppose that for all $t \in [T]$, the Gibbs-measures $\mu_t, \mu_{t+1}$ are $(\Phi, \delta)$-close. If $q_t$ is the marginal distribution induced by \pref{alg:ctrl} on its iterates $x_t$, then we have that 
\begin{itemize}
\item If $B = \infty$, then for all $t$, $\|q_t - \barmu_t\|_\TV \leq 3\delta(t-1)$.
\item If $B=3\ptil T$, then we have 
	$$ \|q_t - \barmu_t\|_{\text{TV}} \leq e^{-\ptil T} + 3\delta(t-1).$$
\end{itemize}	 
\end{lemma}
\begin{proof}
We first consider the $B = \infty$ case. We prove that $\|q_t - \barmu_t\|_\TV \leq 3\delta(t-1)$ by induction on $t$. For $t=1$, the claim is trivially true. So assume it is true for some $t$ and now we prove it for $t+1$.  Let $M = \{x \in \K \mid \Phi^{-1} \leq \frac{\barmu_{t+1}(x)}{\barmu_t(x)} \leq \Phi\}$. Then by \pref{def:phidelta}, we have $\barmu_t(M) \geq 1-\delta$ and $\barmu_{t+1}(M) \geq 1-\delta$. Next, let $\tilde{\mu}_t$ be the distribution of $X \sim \barmu_t$ conditioned on the event $X \in M$. Since $\barmu_t(M) \geq 1-\delta$, it is easy to see that $\|\barmu_t - \tilde{\mu}_t\|_\TV \leq \delta$. Let $\tilde{q}_{t+1}$ be the distribution of $x_{t+1}$ if $x_t$ were sampled from $\tilde{\mu}_t$ instead of $q_t$. Let $E$ be any measurable subset of $\K$. Using the facts that for any $x \in M$, we have $\Pi(\frac{\barmu_{t+1}(x)}{\Phi\barmu_t(x)}) = \frac{\barmu_{t+1}(x)}{\Phi\barmu_t(x)}$, and that $\tilde{\mu}_t(x) = \frac{\barmu_t(x)}{\barmu_t(M)}$, we have
\begin{align*}
\tilde{q}_{t+1}(E) & = \int_{x \in E} \bigg(  \Pr(S'_t=0 | x_t = x)\Pr(x_{t+1} \in E | x_t = x, S'_t=0)  \\
& \qquad \qquad + \Pr((S'_t=1 \land S_t=0)| x_t = x) \Pr(x_{t+1} \in E  |x_t = x, (S'_t=1 \land S_t=0)) \\
 & \qquad \qquad  + \Pr((S'_t=1\land S_t=1)| x_t = x) \Pr(x_{t+1} \in E |x_t = x, (S'_t=1 \land S_t=1)) \bigg) \tilde{\mu}_t(x)dx \\
&= p\barmu_{t+1}(E) +  (1-p)\barmu_{t+1}(E) \int_{M} \left(1-\frac{\barmu_{t+1}(x)}{\Phi\barmu_t(x)}\right)\left(\frac{\barmu_t(x)}{\barmu_t(M)}\right) dx \\
& + (1-p) \int_{E \cap M} \left(\frac{\barmu_{t+1}(x)}{\Phi\barmu_t(x)}\right)\left(\frac{\barmu_t(x)}{\barmu_t(M)}\right) dx\\
&= p\barmu_{t+1}(E) + (1-p) \barmu_{t+1}(E)\left(1 - \frac{\barmu_{t+1}(M)}{\Phi \barmu_t(M)}\right) + (1-p)\frac{\barmu_{t+1}(E \cap M)}{\Phi \barmu_t(M)}.
\end{align*}
Thus,
\begin{align*}
|\tilde{q}_{t+1}(E) - \barmu_{t+1}(E)| &= \frac{1-p}{\Phi \barmu_t(M)} 
|\barmu_{t+1}(E)\barmu_{t+1}(M) - \barmu_{t+1}(E \cap M)| \\
&= \frac{1-p}{\Phi \barmu_t(M)}|\barmu_{t+1}(E \setminus M) - \barmu_{t+1}(E \cap M)\barmu_{t+1}(M^c)|\\
&\leq \frac{\delta}{1-\delta},
\end{align*}
since $\barmu_t(M) \geq 1-\delta$ and $\barmu_{t+1}(M) \geq 1-\delta$. Since $\delta \leq \frac{1}{2}$, we conclude that
\[\|\tilde{q}_{t+1} - \barmu_{t+1}\|_\TV \leq 2\delta.\]
Furthermore, we have 
\[\|q_{t+1} - \tilde{q}_{t+1}\|_\TV \leq \|q_t - \tilde{\mu}_t\|_\TV \leq \|q_t - \barmu_t\|_\TV + \|\barmu_t - \tilde{\mu}_t\|_\TV \leq 3\delta(t-1) + \delta,\]
where the first inequality follows by the data-processing inequality for f-divergences like TV-distance (note that $q_{t+1}$ and $\tilde{q}_{t+1}$ are obtained from $q_t$ and $\tilde{\mu_t}$ respectively via the same data-processing channel), and the second inequality is due to the induction hypothesis. Thus, we conclude that
\[\|q_{t+1} - \barmu_{t+1}\|_\TV \leq \|q_{t+1} - \tilde{q}_{t+1}\|_\TV + \|\tilde{q}_{t+1} - \barmu_{t+1}\|_\TV \leq 3\delta(t-1) + \delta + 2\delta = 3\delta t,\]
completing the induction.

We now turn to the $B = 3\ptil T$ case. Let $q'_t$ be the distribution of $x_t$ if $B = \infty$. We now relate $q'_t$ and $q_t$. We start by defining $q_{\text{all}}$ as the probability distributions over all possible random variables, i.e. $S_{1:T}, S'_{1:T}, Z_{1:T}, x_{1:T}$, sampled by \pref{alg:ctrl}. Similarly, let $q'_{\text{all}}$ be the analogue for the infinite switching budget variant. Let $\E$ be the event that $\sum_{t=1}^T\zeta_t \geq 3\ptil T$. Note that \pref{lem:glm-hp} implies that both $q_{\text{all}}(\E), q_{\text{all}}'(\E) \leq e^{-\ptil T}$. Therefore we have that,
 \begin{align*}
     \|q_{\text{all}} - q'_{\text{all}}\|_{\text{TV}} &=  \sup_{\text{measurable } A} \left( q_{\mathrm{all}}(A) - q_{\mathrm{all}}'(A) \right) \\
     &= \sup_{\text{measurable } A} \left( q_{\mathrm{all}}(A \cap \E) - q_{\mathrm{all}}'(A \cap \E) + \underbrace{q_{\mathrm{all}}(A \cap \neg \E) - q_{\mathrm{all}}'(A \cap \neg \E)}_{=0}\right) \\
     &= \sup_{\text{measurable } A} \left( q_{\mathrm{all}}(A \cap \E) - q_{\mathrm{all}}'(A \cap \E)\right)\\
     &\leq e^{-\ptil T}
 \end{align*}
 Now, for any $t$, since $q_t, q'_t$ are marginals of $q_{\text{all}}, q'_{\text{all}}$ respectively, we have $$\|q_t-q'_t\|_{\text{TV}} \leq \|q_{\text{all}} - q'_{\text{all}}\|_{\text{TV}} \leq e^{-\ptil T}.$$
 Since we have $\|\mu_t - q'_t\|_\TV \leq 3\delta(t-1)$ by the $B = \infty$ analysis, the proof is complete by the triangle inequality.
\end{proof}



 Finally, we restate and prove \pref{thm:regret} here:
 \repeattheorem{thm:regret}
 \begin{proof}
Recall that we defined $\mu_t$ to be the distribution with density proportional as
\[ \mu_t(x) \propto \exp\left( -\beta \left( \sum_{\tau=1}^{t-1} l_{\tau}(x)  + \lambda \cdot \frac{\|x\|^2}{2} \right)\right)\]
Let $q_t$ be the distribution induced by \pref{alg:ctrl} on its iterates $x_t$. \pref{lem:samedist} establishes that the sequence of iterates $x_t$ played by \pref{alg:ctrl} follows $\mu_t$ approximately. We define a sequence of random variables $\{y_t\}$ wherein each $y_t$ is sampled from $\mu_t$ independently. In the following we only prove the case when $B=3\ptil T$, the $B=\infty$ can easily be derived by using the bounds from \pref{lem:samedist} appropriately. We leverage the following lemma,
\begin{lemma}[\cite{levin2017markov}] 
\label{lem:TVExpect}
	For a pair of probability distributions $\mu, \nu$,  each supported on $\K$, we have for any function $f:\K\to \reals$ that
	$$ | \mathbb{E}_{x\sim \mu} f(x) - \mathbb{E}_{x\sim \nu} f(x) | \leq 2 \|\mu-\nu\|_{\text{TV}}\max_{x\in \K} |f(x)|. $$
\end{lemma}

We can now apply \pref{lem:TVExpect} to pair $x_t\sim q_t$ and $y_t\sim \mu_t$, using \pref{lem:samedist}, and functions $\bar{l}_t(x) = l_t(x) - l_t(\bar{x})$, where $\bar{x}\in\K$ is chosen arbitrarily, to arrive at
\begin{equation}\label{eq:mainproofdistshift}
	\left|\mathbb{E} \left[\sum_{t=1}^T (l_t(x_t) - l_t(y_t)) \right]\right| \leq \sum_{t=1}^T \left|\mathbb{E} \left[l_t(x_t)-l_t(y_t)\right]\right| \leq \sum_{t=1}^T \left|\mathbb{E} \left[\bar{l}_t(x_t)-\bar{l}_t(y_t)\right]\right| \leq 2GDT\left(e^{-\ptil T} + 3\delta T \right),
\end{equation} 
where we use that $\max_t \max_{x\in \K} |l_t(x)-l_t(\bar{x})| \leq G \max_{x\in \K}\|x-\bar{x}\| \leq GD$. 
Therefore hereafter we only focus on showing the expected regret bound for the sequence $y_t$. 

We take a distributional approach to the regret bound by defining the function $l^{\Delta}_t: \Delta(\K) \rightarrow \reals$ as $l^{\Delta}_t(\mu) \defeq \bbE_{x \sim \mu} l_t(x)$. We can now redefine the regret in terms of the distributions as follows 
\[\mathrm{Regret}(\mu) = \sum_{t=1}^T l^{\Delta}_t(\mu_t) - \sum_{t=1}^T l^{\Delta}_t(\mu).\]

Let $x^{*} \defeq \arg\min_{x \in \K}\sum_{t=1}^T l_t(x)$. 
Note that $\arg\min_{\mu \in \Delta(\K)}\sum_{t=1}^T l^{\Delta}_t(\mu)$ is the Dirac-delta distribution at $x^*$, and that $\min_{\mu \in \Delta(\K)}\sum_{t=1}^T l^{\Delta}_t(\mu) = \sum_{t=1}^T l_t(x^*)$.
For a given value $\epsilon \in [0, 1]$ define the set $\K_{\epsilon}: \{ \epsilon x + (1-\epsilon) x^* | x \in \K\}$. Let $\mu_{\epsilon}^*$ to be uniform distribution over the set $\K_{\epsilon}$. It is now easy to see using the Lipschitzness of $l_t$,
\begin{equation}
    \label{eqn:mu-delta-approx}
    \sum_{t=1}^T l^{\Delta}_t(\mu^*_{\epsilon}) - \min_{\mu \in \Delta(\K)}\sum_{t=1}^T l^{\Delta}_t(\mu) \leq GDT\epsilon.
\end{equation}
Further we define a proxy loss function $l_0(x) = \frac{\lambda}{2}\|x\|^2$ and correspondingly, $l_0^{\Delta}$. Finally define $\mu_0$ as the uniform distribution over the set $\K$. The following lemma establishes an equivalence between sampling from $\mu_t$
and a Follow-the-regularized-leader strategy in the space of distributions.
\begin{lemma}
\label{lem:mu-t-variational}
Consider an arbitrary distribution $\mu_0$ on $\K$ (referred to as the prior) and $f$ be an arbitrary bounded function on $\K$. Define the distribution $\mu$ over $\K$ with density $\mu(x) \propto \mu_0(x)e^{-(x)}$. Then we have that
    \[\mu = \arg\min_{\mu' \in \Delta(\K)} \left( \bbE_{x \sim \mu'}[f(x)] + \mathrm{KL}(\mu' \| \mu_0) \right).\]
\end{lemma}
The lemma follows from the Gibbs variational principle and a proof is included after the current proof. Using the above lemma, we have that at every step $t \geq 1$,
    \[\mu_t = \min_{\mu \in \Delta(\K)} \left( \sum_{\tau=0}^{t-1} \beta \cdot l_{\tau}^{\Delta}(\mu) + \mathrm{KL}(\mu \| \mu_0) \right).\]
Using the above and the FTL-BTL Lemma (\pref{lem:ftlbtl}) we get the following 
\begin{align*}
    \beta \cdot \left( \sum_{t=1}^{T} \left(l_t^{\Delta}(\mu_t) - l_t^{\Delta}(\mu_{\epsilon}^*) \right) \right) &\leq \beta \cdot \left( \sum_{t=1}^{T} \left( l_t^{\Delta}(\mu_t) - l_t^{\Delta}(\mu_{t+1}) \right) \right) + \beta \cdot (l_0^\Delta(\mu_\epsilon^*) - l_0^\Delta(\mu_1)) \\
    & \qquad + \mathrm{KL}(\mu_{\epsilon}^* \| \mu_0)  {- \mathrm{KL}(\mu_0 \| \mu_0)}\\
    &\leq \beta \cdot \left(\sum_{t=1}^{T} \left( \bbE_{x \sim \mu_t}[\beta \cdot l_t(x)] - \bbE_{x \sim \mu_{t+1}}[\beta \cdot l_t(x)] \right) \right) + \beta \cdot {l_0^\Delta(\mu_\epsilon^*)} + \mathrm{KL}(\mu_{\epsilon}^* \| \mu_0)
\end{align*}
Now using Lemma \ref{lem:waserstein-dist}, there is a coupling $\gamma$ between $\mu_t$ and $\mu_{t+1}$ such that $\sup_{(x, x') \sim \gamma} \|x - x'\| \leq \frac{G}{\lambda}$. Using this coupling we get that,
\begin{align*}
    \sum_{t=1}^{T} \left( \bbE_{x \sim \mu_t}[l_t(x)] - \bbE_{x \sim \mu_{t+1}}[l_t(x)] \right) 
    &= \sum_{t=1}^{T} \bbE_{(x,x') \sim \gamma} [l_t(x) - l_t(x')] \\
    &\leq \sum_{t=1}^{T} \bbE_{(x,x') \sim \gamma} G\|x-x'\| \leq \sum_{t=1}^{T} G^2/\lambda \leq \frac{G^2T}{\lambda}
\end{align*}
Combining the above two displays one gets the following
\begin{align*}
    \mathrm{Regret}(\mu_{\epsilon}^*) &= \sum_{t=1}^T l^{\Delta}_t(\mu_t) - \sum_{t=1}^T l^{\Delta}_t(\mu_{\epsilon}^*) \\
    &\leq l_0^{\Delta}(\mu_{\epsilon}^*) + \frac{G^2T}{\lambda} + \frac{\mathrm{KL}(\mu_{\epsilon}^* \| \mu_0)}{\beta} \\
    &\leq \frac{\lambda D^2}{2} + \frac{G^2T}{\lambda} + \frac{d}{\beta}\log(1/\epsilon)
\end{align*}
where we use that $\mathrm{KL}(\mu_{\epsilon}^* \| \mu_0) = d\log(1/\epsilon)$, since $\mu_{\epsilon}^*$ is the uniform distribution over $\K_{\epsilon} \subseteq \K$ and $\frac{ \mathrm{Vol}(\K_{\epsilon}) }{ \mathrm{Vol}(\K) } = \epsilon^d$. Setting $\epsilon = 1/T$ and using \eqref{eqn:mu-delta-approx} we get that for any $\mu$,
\[ \mathrm{Regret}(\mu) \leq \frac{\lambda D^2}{2} + \frac{G^2T}{\lambda} + \frac{d \log(T)}{\beta} + GD.\]
Combining the above with \eqref{eq:mainproofdistshift} finishes the proof.
\end{proof}
We finish this section with the proof of \pref{lem:mu-t-variational}.
\begin{proof}[Proof of Lemma \ref{lem:mu-t-variational}]
The follows from the following Lemma appearing as in \cite{donsker1975asymptotic}
\begin{lemma}[Lemma 2.1 \cite{donsker1975asymptotic}(rephrased)]
Let $\U$ be the set of continuous functions on $\K$ satisfying $u(x) \in [c_1, c_2]$ for all $u \in \U, x \in \K$, for some constants $c_1, c_2 > 0$. Let $\nu_1$ and $\nu_2$ be any distributions on $\K$, then we have that
\[\mathrm{KL}(\nu_1\|\nu_2) = \sup_{u \in \U} \left( \bbE_{x\sim \nu_1}[\log(u(x))] - \log(\bbE_{x \sim \nu_2}[u(x)]) \right) \]
\end{lemma} 

Using the above lemma, setting $\nu_1 = \mu$, $\nu_2=\mu_0$, $u(x) = e^{-f(x)}$, we get that 
\[ -\log(\bbE_{x \sim \mu_0}[e^{-f(x)}]) \leq \bbE_{x\sim \mu}[f(x)] + \mathrm{KL}(\mu\|\mu_0).\]
Let $Z = \int_{K}e^{-f(x)}\mu_0(x)dx$, then we have that
\[ \bbE_{x\sim\mu}[f(x)] + \mathrm{KL}(\mu\|\mu_0) = \bbE_{x\sim\mu}[f(x)] + \int_{K}\mu(x) \log(e^{-f(x)}/Z)dx = - \log(Z) = -\log(\bbE_{x\sim\mu_0}[e^{-f(x)}]).\]
Combining the above two displays finishes the proof.
\end{proof}

%% file: privacy.tex
\section{Privacy Analysis}
\label{app:privacy}

For brevity of notation, we say two random variables $X, Y$ supported on some set $\Omega$  are $(\varepsilon,\delta)$-indistinguishable if for any outcome set $O\subseteq \Omega$, we have that $$ \Pr(X\in O) \leq e^\varepsilon \Pr(Y\in O) +\delta.$$

We restate and prove \pref{thm:dp}:
\repeattheorem{thm:dp}
\begin{proof}
Consider any two $t$-indexed loss sequences $l_{1:T}, l'_{1:T}\in \L^T$ that differ at not more than one index $t_0\in [T]$, i.e. it is the case that $l_t(x)=l'_t(x)$ holds for all $x\in \K$ and $t\in T-\{t_0\}$. For ease of argumentation we will show differential privacy for the outputs $x_t$ of the algorithm along with the internal variables $\zeta_t$ which are defined for any $t$ in the algorithm as 
\[\zeta_t \defeq \mathbb{I}\{S'_t=0 \text{ or } S_t =0\}.\]

To establish privacy, let $\{(x_t, \zeta_t)\}_{t=1}^{T}$ and $\{(x'_t, \zeta'_t)\}_{t=1}^{T}$ be the instantiations of the random variables determined by \pref{alg:ctrl} upon execution on $l_{1:T}$ and $l'_{1:T}$, respectively. For brevity of notation, we will denote by $\Sigma_t$ the random variable $\{x_{\tau}, \zeta_{\tau}\}_{\tau=1}^t$. We denote by $\mathbf{\Sigma_t}$ all possible values $\Sigma_t$ can take. We now show the following claim, 

\begin{claim}\label{claim:indipriv}
Let $\delta' \geq 0$ and $\Phi$ be as defined in \pref{thm:dp}. Then for any $t\in [T]$ the random variable pairs $(x_t, \zeta_t)$ and $(x'_t, \zeta'_t)$ are $(\varepsilon_t, \delta_t)$-indistinguishable when conditioned on $\Sigma_{t-1}$, i.e. when conditioned on identical values of random choices made by the algorithm before (but not including) round $t$, where $\delta_t=4\delta' + 9\delta'T + 3e^{-\ptil T}$ and 
\begin{align}
    \varepsilon_t = \begin{cases}
    0, & t<t_0\\
    \mathbb{I}_{\sum_{s=1}^{t-1} \zeta_s<B} \cdot 2 \log(\Phi)/p, & t=t_0\\
    \mathbb{I}_{\sum_{s=1}^{t-1} \zeta_s<B}\left(\zeta_{t-1}\log(\Phi) + \frac{2G^2 \beta /{\lambda}}{p} \right)& t>t_0
    \end{cases}
\end{align}
\end{claim}
The proof of the above claim appears after the present proof. 

We intend to use adaptive strong composition for differential privacy (\pref{lem:comp}) with \pref{claim:indipriv} and to that end consider the following calculations
\begin{align*}
\sum_{t=1}^T \varepsilon_t^2 &\leq 
  \frac{4\log^2(\Phi)}{p^2} +  2B\log^2({\Phi}) + \frac{8G^4 \beta^2 /{\lambda^2}}{p^2}T \\
&\leq \frac{4\log^2(\Phi)}{p^2} +  6pT\log^2({\Phi}) + 12\log^3({\Phi})T + \frac{8G^4 \beta^2 /{\lambda^2}}{p^2}T \\
&\text{(Using $B = 3pT + 3(1-\Phi^{-2})T \leq 3pT + 6\log(\Phi)T$)}\\
&= \frac{4\log^2(\Phi)}{p^2} +  3pT\log^2({\Phi}) + 12\log^3({\Phi})T + 3pT\log^2({\Phi}) + \frac{8G^4 \beta^2 /{\lambda^2}}{p^2}T   \\
&\leq 7T^{2/3}\log^2({\Phi}) + 12\log^3({\Phi})T + 11 \left(\frac{G^4 \beta^2}{\lambda^2}\right)^{1/3}\log^{4/3}(\phi) \cdot T \\  \text{ and}\\
\sum_{t=1}^T\delta_t &= 4\delta'T + 9T^2 \delta' + 3Te^{-\ptil T} \leq \frac{\delta}{6} + 3Te^{-pT} + 3Te^{-(1-\Phi^{-2})T} \leq \frac{\delta}{3} + 3Te^{-(1-\Phi^{-2})T}.
\end{align*}
Using the above calculations and applying \pref{lem:comp} with $\delta' = \delta/2$ (in \pref{lem:comp}) concludes the proof.  
\end{proof}

\begin{proof}[Proof Of \pref{claim:indipriv}] We begin by defining a subset $\mathcal{E}_t \in \K$ for all $t$ as 
$$ \mathcal{E}_t = \left\{ x \in \K \bigg| \left( \frac{\barmu_{t+1} (x)}{\Phi \barmu_t(x)} \in \left[ \frac{1}{\Phi^2},  1\right] \right) \land \left( \frac{\barmu'_{t+1} (x)}{\Phi \barmu'_t(x)} \in \left[ \frac{1}{\Phi^2},  1\right] \right) \right\}.$$
The following claim whose proof is presented after the present proof shows that $\E_t$ occurs with high probability conditioned on $\Sigma_{t-1}$ taking any value $\Sigma$ in its domain. 

\begin{claim}\label{claim:highradio}
Let $\Phi$ be as defined in \pref{thm:dp}, then we have that for all $\Sigma \in \mathbf{\Sigma}_t$,
$$ \Pr(x_t \in \mathcal{E}_t | \Sigma_{t-1} = \Sigma) \geq 1 - 3\delta' - 9T\delta' - 3e^{-\ptil T}.$$
\end{claim}

The general recipe we will follow in the proof is to show that $x_t, x'_t$ are $(\varepsilon_x, \delta_x)$-indistinguishable conditioned on $\Sigma_{t-1}$ and the event that $x_t \in \E_t$, for some $(\varepsilon_x, \delta_x)$. We will then show that $\zeta_t, \zeta'_t$ are $(\varepsilon_\zeta, \delta_\zeta)-$indistinguishable after conditioning on $\Sigma_{t-1}, x_t=x$ (and $x'_t=x$ respectively) for an arbitrary $\E_t$. Then, by standard composition of differential privacy \citep{dwork2014algorithmic}, it is implied that $(x_t,\zeta_t)$, $(x'_t, \zeta'_t)$ are $(\varepsilon_x+\varepsilon_\zeta, \delta_x+\delta_\zeta)$ indistinguishable when conditioned on $\Sigma_{t-1}$ and the event that $x_t \in \E_t$. It then follows that the same pair is $(\varepsilon_x+\varepsilon_\zeta, \delta_x+\delta_\zeta+\Pr(x_t \notin \mathcal{E}_t|\Sigma_{t-1}))$ indistinguishable when conditioned on $\Sigma_{t-1}$.

To execute the above strategy, we will examine the three cases -- 
{\em ante} $t<t_0$, {\em at} $t=t_0$, and {\em post} $t>t_0$ -- separately. Recall that $l_{1:T}$ and $l'_{1:T}$ are loss function sequences that differ only at the index $t_0$.  

\paragraph{Ante Case: $t \leq t_0$}: Observe that since $l_{1:t_0-1}=l'_{1:t_0-1}$, having not yet encountered a change (at $t=t_0$) in loss, the algorithm produces identically distributed outputs for the first $t_0$ rounds upon being fed either loss sequence. Therefore we have that
\begin{equation}
\label{eqn:antecase}
    \forall t < t_0, \;\; (x_t, \zeta_t) \text{ and } (x'_t, \zeta'_t) \text{ are } (0,0)-\text{indistinguishable}
\end{equation}

For the remaining two cases, we first assume that number of switches so far have not exceeded $B$, i.e. $\sum_{s=1}^{t-1} \zeta_s = \sum_{s=1}^{t-1} \zeta_s < B$ (conditioned on the same history). If not then both algorithms become deterministic from this point onwards and are $(0,0)$-indistinguishable.


\paragraph{At Case: $t=t_0$:} For the {\em at case}, the last display in the {\em ante} case also means that $x_{t_0}$ and $x'_{t_0}$ are identically distributed random variables. Therefore, to conclude the claim for $t_0$, we need to demonstrate that $\zeta_{t_0}$ and $\zeta'_{t_0}$ are indistinguishable when also additionally conditioned on $x_{t_0}=x'_{t_0}$. We now observe that for any $x \in \E_{t_0}$ and any $\Sigma \in \mathbf{\Sigma_{t_0-1}}$,
\begin{align*}
\frac{\Pr({\zeta'_{t_0}=1}|\Sigma_{t_0-1} = \Sigma, x'_{t_0}=x)}{\Pr({\zeta_{t_0}=1}|\Sigma_{t_0-1} = \Sigma, x_{t_0}=x)} &= \frac{p + (1-p) \left(1-\frac{\barmu'_{{t_0}+1}(x)}{\Phi \barmu'_{{t_0}}(x)}\right)}{p + (1-p) \underbrace{\left(1-\frac{\barmu_{{t_0}+1}(x)}{\Phi \barmu_{{t_0}}(x)}\right)}_{\geq 0}} \\
&\leq \frac{p + (1-p) \left(\underbrace{1-\frac{\barmu'_{{t_0}+1}(x)}{\Phi \barmu'_{{t_0}}(x)}}_{\geq 0}\right)}{p} \\
&\leq 1+ \frac{1}{p} \left(1-\frac{\barmu'_{{t_0}+1}(x)}{\Phi \barmu'_{{t_0}}(x)}\right) \leq 1+ \frac{1}{p} \left(1-\Phi^{-2}\right)\\
&\leq 1+ \frac{1}{p}(1-e^{-2\log \Phi}) \leq 1+ \frac{2\log(\Phi)}{p} \leq e^{2\log \Phi/p},
\end{align*}
using the definition of the set $\E_{t_0}$ and that for any real $x$ $1+x\leq e^x$. Similarly, we have for any $x \in \E_{t_0}$,
\begin{align*}
\frac{\Pr({\zeta'_{t_0}=0}|\Sigma_{t_0-1} = \Sigma, x'_{t_0}=x)}{\Pr({\zeta_{t_0}=0}|\Sigma_{t_0-1} = \Sigma, x_{t_0}=x)} &= \frac{(1-p) \frac{\barmu'_{{t_0}+1}(x)}{\Phi \barmu'_{{t_0}}(x)}}{ (1-p) \frac{\barmu_{{t_0}+1}(x)}{\Phi \barmu_{{t_0}}(x)}} = \frac{\barmu'_{{t_0}+1}(x)}{ \barmu'_{{t_0}}(x)}  \frac{ \barmu_{{t_0}}(x)}{\barmu_{{t_0}+1}(x)}\leq e^{2\log \Phi}.
\end{align*}

The above displays thereby imply that conditioned on $\Sigma_{t_0-1}$ and the event $x_t \in \E_{t_0}$, we have that $(x_{t_0}, \zeta_{t_0})$ and $(x'_{t_0}, \zeta'_{t_0})$ are $(2\log(\Phi)/p, 0)$-indistinguishable. Thereby combining with \pref{claim:highradio} we get that conditioned on $\Sigma_{t-1}$
\begin{equation}
\label{eqn:atcase}
      (x_{t_0}, \zeta_{t_0}) \text{ and } (x'_{t_0}, \zeta'_{t_0}) \text{ are } (2\log(\Phi)/p, 3\delta' + 9T\delta' + 3e^{-\ptil T})-\text{indistinguishable}
\end{equation}

\paragraph{Post Case: $t > t_0$:} Recall that while claiming indistinguishability of appropriate pair of random variables, we condition on a shared past of $\Sigma_{t-1}$. In particular, this means that $x'_{t-1}=x_{t-1}$ and that $\zeta_{t-1}=\zeta'_{t-1}$. Now, if $\zeta_{t-1}=0$, then $x'_t = x'_{t-1} = x_{t-1}=x_t$. If $\zeta_{t-1}=1$, the iterates are sampled as $x_t \sim \barmu_t$ and $x'_t\sim \barmu'_t$ in round $t$. Once again by applying the condition on $\Phi$ as stated in \pref{thm:dp} we have that $x_t, x_{t}'$ are $( \zeta_{t-1} \log \Phi,\delta')$-indistinguishable.


To conclude the claim and hence the proof, we need to establish the indistinguishability of $\zeta_t$ and $\zeta'_t$ conditioned additionally on the event $x_t=x'_t$. Unlike for $t=t_0$, the analysis here for $\zeta$'s is more involved. To proceed, we first obtain a second-order perturbation result. We have
\begin{align*}
   \frac{\barmu_{t+1}(x)}{\barmu_{t}(x)} &= \frac{\exp\left( -\beta \left( l_{1:t}(x) + \frac{\lambda}{2} \|x\|^2 \right) \right)}{\exp\left( -\beta \left( l_{1:t-1}(x) + \frac{\lambda}{2} \|x\|^2 \right) \right)} \cdot \frac{\int_{x \in \K} \exp\left( -\beta \left( l_{1:t}(x) + \frac{\lambda}{2} \|x\|^2 \right) \right)dx}{ \int_{x \in \K} \exp\left( -\beta \left( l_{1:t-1}(x) + \frac{\lambda}{2} \|x\|^2 \right) \right)dx} \\
   &\defeq \exp(-\beta \cdot l_t(x)) \cdot \frac{Z(l_{1:t-1})}{Z(l_{1:t})}
\end{align*}
where we have defined $Z(l) = \int_{x \in \K} \exp\left( -\beta \left( l(x) + \frac{\lambda}{2} \|x\|^2 \right) \right)dx$. Define $B_t = \frac{Z(l_{1:t-1})}{Z(l_{1:t})}$. To bound $B_t$ we define the following scalar function $p(t):[0,1] \rightarrow \reals$ as $p(t) = \log(Z(l_{1:t-1} + t \cdot l_t), \beta, \lambda)$. The following lemma shows that $p(t)$ is a convex function and characterizes the derivative of $p$.
\begin{lemma}
\label{lem:log-partition-lemma}
    Given two differentiable loss functions $f,g$, and any number $t \in \reals$ define the measure $\mu(t)(x)$ over a convex set $\K$ as $\mu(t) = \exp( - (f(x) + tg(x)))$. Further define the log partition function of $\mu(t)$, $p(t) \defeq \log \left( \int_{x \in K} \exp( - (f(x) + tg(x)) dx \right)$. Define the probability disitrbution $\barmu(t)(x) = \frac{\mu(t)(x)}{p(t)}$. We have that $p(t)$ is a convex function of $t$. Futhermore $p'(t) = \bbE_{x \sim \barmu(t)}[-g(x)]$.
\end{lemma}

\begin{proof}[Proof of \pref{lem:log-partition-lemma}]
We first prove the derivative. Consider the following calculation
\[p'(t) = \frac{\int_{x \in K} -g(x) \cdot \exp( - (f(x) + tg(x)) dx}{\int_{x \in K} \exp( - (f(x) + tg(x)) dx} = \bbE_{x \sim \barmu(t)}[-g(x)]\]
To prove convexity we consider $p''(t)$. Once again, we can calculate as follows:
\begin{align*}
    p''(t) & = \frac{\int_{x \in K} g^2(x) \cdot \exp( - (f(x) + tg(x)) dx}{\int_{x \in K} \exp( - (f(x) + tg(x)) dx} - \left(\frac{\int_{x \in K} g(x) \cdot \exp( - (f(x) + tg(x)) dx}{\int_{x \in K} \exp( - (f(x) + tg(x)) dx} \right)^2 \\
    &= \mathrm{Var}_{\barmu(t)}(g(x)) \geq 0.
\end{align*}
Since $p''(t) \geq 0$ this proves that the function is convex.
\end{proof}
In particular using the above lemma we get that 
\[ \log(B_t) = p(0) - p(1) \leq -\frac{\partial p(0)}{\partial t} = \bbE_{y \sim \barmu_{t}}[\beta \cdot l_t(y)]\]
\[ \log(B_t) = p(0) - p(1) \geq -\frac{\partial p(1)}{\partial t} = \bbE_{y \sim \barmu_{t+1}}[\beta \cdot l_t(y)] \]
 It now follows that 
\begin{align*}
        \log \frac{\barmu_{t+1}(x)}{\barmu_{t}(x)} &\leq -\beta \cdot l_t(x) + \bbE_{y \sim \barmu_{t}}[\beta \cdot l_t(y)]\\
       \log \frac{\barmu_{t+1}(x)}{\barmu_{t}(x)} &\geq -\beta \cdot l_t(x) + \bbE_{y \sim \barmu_{t+1}}[\beta \cdot l_t(y)].
\end{align*}
Similarly for $\barmu'$, one can establish
\begin{align*}
        \log \frac{\barmu'_{t+1}(x)}{\barmu'_{t}(x)} &\leq -\beta \cdot l'_t(x) + \bbE_{y \sim \barmu'_{t}}[\beta \cdot l'_t(y)]\\
       \log \frac{\barmu'_{t+1}(x)}{\barmu'_{t}(x)} &\geq -\beta \cdot l'_t(x) + \bbE_{y \sim \barmu'_{t+1}}[\beta \cdot l'_t(y)].
\end{align*}
At this point, note that since $t>t_0$, $l'_t=l_t$, and that $l_{1:t-1}-l'_{1:t-1} = l_{t_{0}}-l'_{t_0}$, we can now bound the term of interest for privacy for all $x$.
\begin{align*}
\log \frac{\frac{\barmu'_{t+1}(x)}{\Phi \barmu'_{t}(x)}}{\frac{\barmu_{t+1}(x)}{\Phi \barmu_{t}(x)}} \leq \;\; \bbE_{y \sim \barmu'_{t}}[\beta \cdot l_t(y)] -  \bbE_{y \sim \barmu_{t+1}}[\beta \cdot l_t(y)].
\end{align*}
Now using \pref{lem:waserstein-dist} twice we get that $W_{\infty}(\barmu'_t, \barmu_{t+1}) \leq \frac{2G}{\lambda}$ which implies that there is a coupling $\gamma$ between $\barmu'_t$ and $\barmu'_{t+1}$ such that $sup_{(y,y') \sim \gamma} \|y-y'\| \leq \frac{2G}{\lambda}$. Therefore we have that 
\[\bbE_{y \sim \barmu'_{t}}[\beta \cdot l_t(y)] -  \bbE_{y \sim \barmu_{t+1}}[\beta \cdot l_t(y)] = \beta \cdot \bbE_{(y,y') \sim \gamma} [l_t(y) - l_t(y')] \leq \beta \cdot G \cdot \bbE_{(y,y') \sim \gamma} [|y-y'|] \leq \frac{\beta \cdot 2G^2}{\lambda}.\]
The above display immediately gives that for all $\Sigma \in \mathbf{\Sigma_{t-1}}$ and for all $x \in \E_{t}$,
\begin{align*}
    \frac{\Pr({\zeta'_{t}=0}|\Sigma'_{t-1} = \Sigma,x'_t=x)}{\Pr({\zeta_{t}=0}|\Sigma_{t-1} = \Sigma, x_t=x)} =  \frac{(1-p)\frac{\barmu'_{t+1}(x)}{\Phi \barmu'_{t}(x)}}{(1-p)\frac{\barmu_{t+1}(x)}{\Phi \barmu_{t}(x)}} \leq e^{\frac{2G^2 \beta}{\lambda}}.
\end{align*}
Now, for the remaining possibility, we have 
 \begin{align*}
 \frac{\Pr({\zeta'_{t}=1}|\Sigma'_{t-1} = \Sigma,x'_t=x)}{\Pr({\zeta_{t}=1}|\Sigma_{t-1} = \Sigma,x_t=x)}&=  \frac{p + (1-p) \left(1-\frac{\barmu'_{t+1}(x)}{\Phi \barmu'_{t}(x)}\right)}{p + (1-p) \left(1-\frac{\barmu_{t+1}(x)}{\Phi \barmu_{t}(x)}\right)} \\
 &\leq \frac{p + (1-p) \left(1-\frac{\barmu_{t+1}(x)}{\Phi \barmu_{t}(x)} e^{-\frac{2G^2 \beta}{\lambda}}\right)}{p + (1-p) \left(1-\frac{\barmu_{t+1}(x)}{\Phi \barmu_{t}(x)}\right)} \\
 &\leq 1 +\frac{  \underbrace{\frac{\barmu_{t+1}(x)}{\Phi \barmu_{t}(x)}}_{\leq 1} \left(1-e^{-\frac{2G^2 \beta}{\lambda}}\right)  }{p} \\
 &\leq e^{\frac{1}{p} \cdot \frac{2G^2 \beta}{\lambda}}.
\end{align*}

The above displays thereby imply that conditioned on $\Sigma_{t-1}$ and $x_t \in \E_{t}$ we have that $\zeta_{t}$ and $\zeta'_{t}$ are $(\frac{2G^2 \beta /{\lambda}}{p} , 0)$-indistinguishable. Thereby we get that conditioned on $\Sigma_{t-1}$

\begin{equation}
\label{eqn:postcase}
      (x_{t}, \zeta_{t}) \text{ and } (x'_{t}, \zeta'_{t}) \text{ are } \left(\zeta_{t-1}\log \Phi + \frac{2G^2 \beta /{\lambda}}{p}, 4\delta' - 9T\delta' - 3e^{-\ptil T} \right)-\text{indistinguishable}
\end{equation}

Combining the statements in Equations \eqref{eqn:antecase}, \eqref{eqn:atcase} and \eqref{eqn:postcase} finishes the proof.

\end{proof}

\begin{proof}[Proof Of \pref{claim:highradio}]
Let $q_t$ be the probability distribution induced on the iterates chosen by \pref{alg:ctrl} when run on a loss sequence $l_{1:T}$. Using the conditions in the theorem and by \pref{lem:samedist}, we have that $\|\barmu_t-q_t\|\leq e^{-\ptil T} + 3T\delta'$ for any $t\in [T]$. From this, noting that $l_{1:t}-l_{1:t-1}$ is $G$-Lipschitz and $\beta$-smooth, we have that for all $t$,
\[\Pr_{X \sim q_t}\left[\frac{1}{\sqrt{\Phi}} \leq \frac{\barmu_{t+1}(X)}{\barmu_t(X)} \leq \sqrt{\Phi}\right] \geq 1 - \delta' - 3T\delta'-e^{-\ptil T} \]
Furthermore noting that $l_{1:t-1}-l'_{1:t-1}$ is $2G$-Lipschitz and $2\beta$-smooth we have that for all $t$,
\[\Pr_{X \sim q_t}\left[\frac{1}{\sqrt{\Phi}} \leq \frac{\barmu_{t}(X)}{ \barmu'_t(X)} \leq \sqrt{\Phi}\right] \geq 1 - \delta' - 3T\delta'-e^{-\ptil T} \]

Similarly noting that $l'_{1:t}-l_{1:t-1}$ is $3G$-Lipschitz and $2\beta$-smooth we can apply the same argument to obtain
\[\Pr_{X \sim q_t}\left[\frac{1}{\sqrt{\Phi}} \leq \frac{\barmu_{t+1}'(X)}{\barmu_t(X)} \leq \sqrt{\Phi}\right] \geq 1 - \delta' - 3T\delta'-e^{-\ptil T} \]
The above statements imply the claim.
\end{proof}






%% file: main.bbl
\begin{thebibliography}{22}
\providecommand{\natexlab}[1]{#1}
\providecommand{\url}[1]{\texttt{#1}}
\expandafter\ifx\csname urlstyle\endcsname\relax
  \providecommand{\doi}[1]{doi: #1}\else
  \providecommand{\doi}{doi: \begingroup \urlstyle{rm}\Url}\fi

\bibitem[Agarwal and Singh(2017)]{AS17}
N.~Agarwal and K.~Singh.
\newblock The price of differential privacy for online learning.
\newblock In \emph{ICML}, volume~70 of \emph{Proceedings of Machine Learning
  Research}, pages 32--40. {PMLR}, 2017.
\newblock URL \url{http://proceedings.mlr.press/v70/agarwal17a.html}.

\bibitem[Agarwal et~al.(2023)Agarwal, Kale, Singh, and
  Thakurta]{pmlr-v195-agarwal23d}
N.~Agarwal, S.~Kale, K.~Singh, and A.~Thakurta.
\newblock Differentially private and lazy online convex optimization.
\newblock In G.~Neu and L.~Rosasco, editors, \emph{Proceedings of Thirty Sixth
  Conference on Learning Theory}, volume 195 of \emph{Proceedings of Machine
  Learning Research}, pages 4599--4632. PMLR, 12--15 Jul 2023.
\newblock URL \url{https://proceedings.mlr.press/v195/agarwal23d.html}.

\bibitem[Altschuler and Talwar(2021)]{AltschulerT21}
J.~M. Altschuler and K.~Talwar.
\newblock Online learning over a finite action set with limited switching.
\newblock \emph{Math. Oper. Res.}, 46\penalty0 (1):\penalty0 179--203, 2021.
\newblock \doi{10.1287/moor.2020.1052}.
\newblock URL \url{https://doi.org/10.1287/moor.2020.1052}.

\bibitem[Anava et~al.(2015)Anava, Hazan, and Mannor]{anava2015online}
O.~Anava, E.~Hazan, and S.~Mannor.
\newblock Online learning for adversaries with memory: price of past mistakes.
\newblock In \emph{Advances in Neural Information Processing Systems}, pages
  784--792, 2015.

\bibitem[Asi et~al.(2023)Asi, Feldman, Koren, and Talwar]{asi2023private}
H.~Asi, V.~Feldman, T.~Koren, and K.~Talwar.
\newblock Private online prediction from experts: Separations and faster rates.
\newblock In \emph{The Thirty Sixth Annual Conference on Learning Theory},
  pages 674--699. PMLR, 2023.

\bibitem[Bakry and {\'E}mery(2006)]{bakry2006diffusions}
D.~Bakry and M.~{\'E}mery.
\newblock Diffusions hypercontractives.
\newblock In \emph{S{\'e}minaire de Probabilit{\'e}s XIX 1983/84: Proceedings},
  pages 177--206. Springer, 2006.

\bibitem[Donsker and Varadhan(1975)]{donsker1975asymptotic}
M.~D. Donsker and S.~S. Varadhan.
\newblock Asymptotic evaluation of certain markov process expectations for
  large time, i.
\newblock \emph{Communications on Pure and Applied Mathematics}, 28\penalty0
  (1):\penalty0 1--47, 1975.

\bibitem[Dwork and Roth(2014)]{dwork2014algorithmic}
C.~Dwork and A.~Roth.
\newblock The algorithmic foundations of differential privacy.
\newblock \emph{Foundations and Trends in Theoretical Computer Science},
  9\penalty0 (3--4):\penalty0 211--407, 2014.

\bibitem[Ganesh et~al.(2023)Ganesh, Thakurta, and
  Upadhyay]{pmlr-v195-ganesh23a}
A.~Ganesh, A.~Thakurta, and J.~Upadhyay.
\newblock Universality of langevin diffusion for private optimization, with
  applications to sampling from rashomon sets.
\newblock In G.~Neu and L.~Rosasco, editors, \emph{Proceedings of Thirty Sixth
  Conference on Learning Theory}, volume 195 of \emph{Proceedings of Machine
  Learning Research}, pages 1730--1773. PMLR, 12--15 Jul 2023.
\newblock URL \url{https://proceedings.mlr.press/v195/ganesh23a.html}.

\bibitem[Geulen et~al.(2010)Geulen, V{\"{o}}cking, and Winkler]{GeulenVW10}
S.~Geulen, B.~V{\"{o}}cking, and M.~Winkler.
\newblock Regret minimization for online buffering problems using the weighted
  majority algorithm.
\newblock In A.~T. Kalai and M.~Mohri, editors, \emph{COLT}, pages 132--143.
  Omnipress, 2010.
\newblock URL
  \url{http://colt2010.haifa.il.ibm.com/papers/COLT2010proceedings.pdf\#page=140}.

\bibitem[Gopi et~al.(2022)Gopi, Lee, and Liu]{pmlr-v178-gopi22a}
S.~Gopi, Y.~T. Lee, and D.~Liu.
\newblock Private convex optimization via exponential mechanism.
\newblock In P.-L. Loh and M.~Raginsky, editors, \emph{Proceedings of Thirty
  Fifth Conference on Learning Theory}, volume 178 of \emph{Proceedings of
  Machine Learning Research}, pages 1948--1989. PMLR, 02--05 Jul 2022.
\newblock URL \url{https://proceedings.mlr.press/v178/gopi22a.html}.

\bibitem[Hazan(2016)]{hazan2016introduction}
E.~Hazan.
\newblock Introduction to online convex optimization.
\newblock \emph{Foundations and Trends in Optimization}, 2\penalty0
  (3-4):\penalty0 157--325, 2016.

\bibitem[Jain et~al.(2012)Jain, Kothari, and Thakurta]{JKT12}
P.~Jain, P.~Kothari, and A.~Thakurta.
\newblock Differentially private online learning.
\newblock In \emph{Proc. of the 25th Annual Conf. on Learning Theory (COLT)},
  volume~23, pages 24.1--24.34, June 2012.

\bibitem[Kairouz et~al.(2021)Kairouz, McMahan, Song, Thakkar, Thakurta, and
  Xu]{kairouz2021practical}
P.~Kairouz, B.~McMahan, S.~Song, O.~Thakkar, A.~Thakurta, and Z.~Xu.
\newblock Practical and private (deep) learning without sampling or shuffling.
\newblock In \emph{ICML}, 2021.

\bibitem[Kalai and Vempala(2005)]{kalai2005efficient}
A.~Kalai and S.~Vempala.
\newblock Efficient algorithms for online decision problems.
\newblock \emph{Journal of Computer and System Sciences}, 71\penalty0
  (3):\penalty0 291--307, 2005.

\bibitem[Ledoux(1999)]{ledoux1999concentration}
M.~Ledoux.
\newblock Concentration of measure and logarithmic sobolev inequalities.
\newblock In \emph{Seminaire de probabilites XXXIII}, pages 120--216. Springer,
  1999.

\bibitem[Levin and Peres(2017)]{levin2017markov}
D.~A. Levin and Y.~Peres.
\newblock \emph{Markov chains and mixing times}, volume 107.
\newblock American Mathematical Soc., 2017.

\bibitem[Merhav et~al.(2002)Merhav, Ordentlich, Seroussi, and
  Weinberger]{MerhavOSW02}
N.~Merhav, E.~Ordentlich, G.~Seroussi, and M.~J. Weinberger.
\newblock On sequential strategies for loss functions with memory.
\newblock \emph{{IEEE} Trans. Inf. Theory}, 48\penalty0 (7):\penalty0
  1947--1958, 2002.
\newblock \doi{10.1109/TIT.2002.1013135}.
\newblock URL \url{https://doi.org/10.1109/TIT.2002.1013135}.

\bibitem[Sherman and Koren(2021)]{sherman2021lazy}
U.~Sherman and T.~Koren.
\newblock Lazy oco: Online convex optimization on a switching budget.
\newblock In \emph{Conference on Learning Theory}, pages 3972--3988. PMLR,
  2021.

\bibitem[Sherman and Koren(2023)]{sherman2021lazyx}
U.~Sherman and T.~Koren.
\newblock Lazy oco: Online convex optimization on a switching budget.
\newblock \emph{arXiv preprint arXiv:2102.03803 version 7 (see also version
  5)}, 2023.
\newblock URL \url{https://arxiv.org/abs/2102.03803}.

\bibitem[Smith and Thakurta(2013)]{st13}
A.~Smith and A.~Thakurta.
\newblock (nearly) optimal algorithms for private online learning in
  full-information and bandit settings.
\newblock In \emph{Advances in Neural Information Processing Systems}, pages
  2733--2741, 2013.

\bibitem[Whitehouse et~al.(2022)Whitehouse, Ramdas, Rogers, and
  Wu]{whitehouse2022fully}
J.~Whitehouse, A.~Ramdas, R.~Rogers, and Z.~S. Wu.
\newblock Fully adaptive composition in differential privacy.
\newblock \emph{arXiv preprint arXiv:2203.05481}, 2022.

\end{thebibliography}
